\documentclass[12pt]{amsart}

\usepackage{amsmath,amsthm,amssymb,amsfonts,verbatim}
\usepackage{graphicx}
\usepackage[colorlinks,
            linkcolor=blue,
            anchorcolor=blue,
            citecolor=blue
            ]{hyperref}
\usepackage[hmargin=1.2in,vmargin=1.2in]{geometry}

\title[Asymptotic construction of LRC]{Construction of asymptotically good locally repairable codes via automorphism groups of function fields}
\author{Xudong Li}\address{Lab of Security Insurance Cyberspace  and School of Science, Xihua University, Chengdu, China
610039}\email{lixudong73@163.com}
\author{Liming Ma}\address{School of Mathematical Sciences, Yangzhou University, Yangzhou, China
225002}\email{lmma@yzu.edu.cn}
\author{Chaoping Xing} \address{Division of Mathematical Sciences, School of Physical Mathematical Sciences,
Nanyang Technological University, Singapore
637371}\email{xingcp@ntu.edu.sg}
\date{}

\newtheorem{lemma}{Lemma}[section]
\newtheorem{theorem}[lemma]{Theorem}
\newtheorem{cor}[lemma]{Corollary}

\newtheorem{prop}[lemma]{Proposition}

\newtheorem{defn}{Definition}

\theoremstyle{remark}
\newtheorem{rmk}{Remark}

\renewcommand{\epsilon}{\varepsilon}
\renewcommand{\le}{\leqslant}
\renewcommand{\ge}{\geqslant}

\def\Gal{{\rm Gal}}

\newcommand{\vnote}[1]{}

\def\F{\mathbb{F}}

\def \mA {\mathcal{A}}

\def \mA {\mathcal{A}}

\def \mL {\mathcal{L}}

\def \mP {\mathcal{P}}

\def \Xi {{X^{[i]}}}

\newcommand{\Ga}{\alpha}
\newcommand{\Gb}{\beta}

\newcommand{\Gs}{\sigma}

\def \bc {{\bf c}}

\def\supp {{\rm supp }}

\def\mG{{\mathcal G}}

\def\Aut {{\rm Aut }}

\def\LRC {{\rm locally repairable code\ }}
\def\LRCs {{\rm locally repairable codes\ }}

\def\Gal{{\rm Gal}}

\begin{document}

\maketitle

\begin{abstract} 
Locally repairable codes have been investigated extensively in recent years due to practical application in distributed storage as well as theoretical interest. However, not much work on asymptotical behavior of locally repairable codes has been done until now. In particular, there is a little result on constructive  lower bound on asymptotical behavior of locally repairable codes. In this paper, we extend  the construction given in  \cite{BTV17} via automorphism groups of function field towers. The main advantage of our construction is to allow more flexibility of locality. Furthermore, we show that the Gilbert-Varshamov type bound  on locally repairable codes can be improved for all sufficiently large $q$.
\end{abstract}

\section{Introduction}
Due to applications in distributed storage systems, locally repairable codes (or locally recoverable codes) have received a lot of attention recently. Most of the papers on this topic focus on  constructions of   locally repairable codes of finite length or various bounds \cite{CM15, GHSY12, PD14, PKLK12, SRKV13, TB14, TBF16, TPD16}. Unlike in the classical coding case, few papers study the asymptotical behavior of locally repairable codes. The main purpose of this paper is to present a construction of asymptotically good locally repairable codes via automorphism groups of function field towers.

\subsection{Locally repairable codes and some bounds}\label{subsec:2.1}
Informally speaking, a block code is said with locality $r$ if  every coordinate of a given codeword can be recovered by accessing at most $r$ other coordinates of this codeword.  Let $q$ be a prime power and let $\F_q$ be the finite field with $q$ elements. The formal definition of a locally repairable code with locality $r$ is given as follows.

\begin{defn}
Let $C\subseteq \F_q^n$ be a $q$-ary block code of length $n$. For each $\Ga\in\F_q$ and $i\in \{1,2,\cdots, n\}$, define $C(i,\Ga):=\{\bc=(c_1,\dots,c_n)\in C\; : \; c_i=\Ga\}$. For a subset $I\subseteq \{1,2,\cdots, n\}\setminus \{i\}$, we denote by $C_{I}(i,\Ga)$ the projection of $C(i,\Ga)$ on $I$.
Then $C$ is called a locally repairable code with locality $r$ if, for every $i\in \{1,2,\cdots, n\}$, there exists a subset
$I_i\subseteq \{1,2,\cdots, n\}\setminus \{i\}$ with $|I_i|\le r$ such that  $C_{I_i}(i,\Ga)$ and $C_{I_i}(i,\Gb)$ are disjoint for any $\Ga\neq \Gb\in \F_q$.
\end{defn}

Apart from the usual parameters: length, rate and minimum distance,  the locality of a   locally repairable code plays a crucial role. In this paper, we always consider locally repairable codes that are linear over $\F_q$. Thus, a $q$-ary \LRC of length $n$, dimension $k$, minimum distance $d$ and locality $r$ is said to be an $[n,k,d]_q$-\LRC with locality $r$.

The well-known Singleton-type bound \cite{GHSY12,PD14} for locally repairable codes with locality $r$  was given by
\begin{equation}\label{Singletonbound}
d\le n-k-\left\lceil \frac{k}{r} \right\rceil+2.
\end{equation}
If we ignore the minimum distance of a $q$-ary locally repairable code, then there is a constraint on the rate \cite{GHSY12}, namely,
\begin{equation}\label{eq:x2}
\frac{k}{n} \le \frac{r}{r+1}.
\end{equation}
In this paper, the minimum distance of a \LRC is taken into consideration. For an $[n,k,d]$-linear code, $k$ information symbols can recover the whole codeword. Thus, the locality $r$ is usually upper bounded by $k$. If we allow $r=k$, i.e., there is no constraint on locality, then the bound \eqref{Singletonbound} becomes the usual Singleton bound that shows constraint on $n,k$ and $d$ only. The other extreme case is that the locality $r$ is $1$. In this case, the \LRC is a repetition code by repeating each symbol twice and the bound \eqref{Singletonbound} becomes $d(C)\le n-2k+2$ which shows  the Singleton bound for repetition codes.

When studying the asymptotical behavior of \LRCs, we only consider the fixed locality $r$ (i.e., $r$ does not change when the length $n$ tends to $\infty$), while the dimension and minimum distance are propositional to the length $n$.
Let $R_q(r,\delta)$ denote the asymptotic bound on the rate of $q$-ary locally repairable codes with locality $r$ and relative minimum distance $\delta$, i.e., $$R_q(r,\delta)=\limsup_{n\rightarrow \infty} \frac{\log_qM_q(n,\lfloor\delta n\rfloor,r)}{n} ,$$
where $M_q(n,d,r)$ is the maximum size of  \LRCs of length $n$, minimum distance $d$ and locality $r$.
Then the Singleton-type bound \eqref{Singletonbound} gives
\begin{equation}\label{S_bound}
R_q(r,\delta)\le  \frac{r}{r+1}(1-\delta) \mbox{ for } 0\le \delta\le 1.
\end{equation}
The following two upper bounds can be  found in \cite{CM15,TBF16}:
\begin{equation}\label{P_bound}
R_q(r,\delta)\le  \frac{r}{r+1}\left(1-\frac{q}{q-1}\cdot \delta\right) \text{ for } 0\le \delta\le 1-q^{-1}
\end{equation}
and
\begin{equation}\label{LP_bound}
R_q(r,\delta)\le \min_{0\le \tau \le \frac{1}{r+1}} \left\{\tau r+(1-\tau (r+1))f_q\left(\frac{\delta}{1-\tau(r+1)}\right)\right\},
\end{equation}
where $f_q(x):=H_q\left(\frac{1}{q}( q-1-x(q-2)-2\sqrt{(q-1)x(1-x)})\right)$ and $H_q(x)$ is the $q$-ary entropy function defined by
$$H_q(x):=x\log_q(q-1)-x\log_q(x)-(1-x)\log_q(1-x).$$
The bound given in \eqref{P_bound} is derived from the Plotkin bound, while the bound given in \eqref{LP_bound} is obtained from the linear programming bound for $q$-ary codes \cite{Aa77}.

For $0\le \delta \le 1-q^{-1}$, the asymptotic Gilbert-Varshamov bound of codes without  locality constraint is given by
$R\ge 1-H_q(\delta)$, and
the asymptotic Gilbert-Varshamov bound of codes with locality constraint is given in  \cite{TBF16}
\begin{equation}\label{asymp_GV_bound}
R_q(r,\delta)\ge 1-\min_{0<s\le 1} \left\{\frac{1}{r+1} \log_q((1+(q-1)s)^{r+1}+(q-1)(1-s)^{r+1})-\delta \log_qs\right\}.
\end{equation}

\subsection{Known results}
Although there are several asymptotically upper bounds and  the asymptotic Gilbert-Varshamov bound on locally repairable codes,   there is little work on asymptotical lower bounds  that are constructive.
By using optimal Garcia-Stichtenoth function field tower,
Barg {\it et al.} \cite{BTV17} gave a construction of  asymptotically good $q$-ary locally repairable codes with locality $r$  whose rate $R$ and relative distance $\delta$ satisfy
\begin{equation}\label{construction_l}
R \ge\frac{r}{r+1}\Big{(}1-\delta-\frac{3}{\sqrt{q}+1}\Big{)},\quad r=\sqrt{q}-1,
\end{equation}
\begin{equation}\label{construction_l+1}
R \ge \frac{r}{r+1}\Big{(}1-\delta-\frac{\sqrt{q}+r}{q-1}\Big{)},\quad (r+1)|(\sqrt{q}+1).
\end{equation}
It was further shown in \cite{BTV17} that for some values $r$ and $q$, the bound \eqref{construction_l+1} exceeds the asymptotic Gilbert-Varshamov bound \eqref{asymp_GV_bound}.

\subsection{Our results} In this paper, we also employ the optimal Garcia-Stichtenoth function field tower \cite{GS96} to construct  asymptotically good locally repairable codes. Our method can be viewed as an extension of the construction in  \cite{BTV17} in the sense that we make use of automorphism groups of the Garcia-Stichtenoth function field tower \cite{La13}. Furthermore, our construction allows more flexibility of locality. More precisely speaking, we have the following main result in this paper.
\begin{theorem}\label{thm:1.1}
Let $q=\ell^2$, with $\ell=p^w$ for a prime $p$ and an integer $w\ge 1$. For any integer $v$ with $0\le v\le w$ and positive integer $u$ satisfying $u|\gcd(p^v-1, \ell -1)$, let $r+1=up^v$. Then there exists a family of explicit $q$-ary linear locally repairable codes with locality $r$ whose rate $R$ and relative distance $\delta$ satisfy
\begin{equation}\label{eq:main}R\ge \frac{r}{1+r}\Big{(}1-\delta-\frac{\sqrt{q}+r-1}{q-\sqrt{q}}\Big{)}.\end{equation}
\end{theorem}
It seems that there are still some constraints on locality $r$ in Theorem \ref{thm:1.1}. However, it allows more flexibility of locality compared with bounds \eqref{construction_l} and \eqref{construction_l+1}. For instance,  the following corollary shows flexible locality.
\begin{cor} If $q$ is a square and it is a power of a prime $p$, then there exists a family of explicit $q$-ary linear locally repairable codes with locality $r$ whose rate $R$ and relative distance $\delta$ achieve the bound {\rm \eqref{eq:main}} provided that $r$ satisfies one of the following conditions
\begin{itemize}
\item[{\rm (i)}] $r+1$ divides $\sqrt{q}$;
\item[{\rm (ii)}] $r+1$ divides $\sqrt{q}-1$;
\item[{\rm (iii)}] $r=up^v-1$ for any divisor $u$ of $p-1$ and $1\le v\le w$;
\item[{\rm (iv)}] $r=u\sqrt{q}-1$ for any divisor $u$ of $\sqrt{q}-1$. 
\end{itemize}
\end{cor}
\begin{proof} Taking $u=1$  in Theorem \ref{thm:1.1} gives the result of (i). Taking $v=0$  in Theorem \ref{thm:1.1} gives the result of (ii). The result of (iii) follows from Theorem \ref{thm:1.1}, since we have $u|\gcd(p^v-1, \ell-1)$ for any $1\le v\le w$. Part (iv) follows from Theorem \ref{thm:1.1} by setting $v=w$.
\end{proof}

\subsection{Comparison} 
Let us first compare our result with the one given in \cite{BTV17}. As our result is not comparable with the bound \eqref{construction_l+1} due to different locality regime, we can only compare with \eqref{construction_l}. It is easy to see that our bound in Theorem \ref{thm:1.1} gives a better result than \eqref{construction_l} when the locality $r$ is $\sqrt{q}-1$. On the other hand, as we mentioned in the previous subsection, the main advantage of this paper is to allow flexible locality.

Next, we compare our result with the asymptotic Gilbert-Varshamov bound \eqref{asymp_GV_bound}. First of all, we give some numerical comparison.  With the help of Sage, we can show that the bound given in  \eqref{eq:main} is better than the  asymptotic Gilbert-Varshamov bound given in \eqref{asymp_GV_bound} if $q,\delta$ and $r$ take the following values:
\begin{itemize}
\item[(1)] $q=2^8$, $\delta=0.5$ and $r\in \{1,2 \};$
\item[(2)] $q=2^{10}$, $\delta=0.5$ and $r\in \{1,3,7,15,30,31 \};$
\item[(3)] $q=2^{12}$, $\delta=0.5$ and $r\in \{1,2,3,6,7,8,11,15,20,31,47,55,62,63 \};$
\item[(4)] $q=3^6$, $\delta=0.5$ and $r\in \{1,2,5,8,12,17,25,26 \};$
\item[(5)] $q=3^8$, $\delta=0.5$ and $r\in \{1,2,3,4,5,7,8,9,15,17,19,26,35,39,53,71,79,80 \};$
\item[(6)] $q=5^4$, $\delta=0.5$ and $r\in \{1,2,3,4,5,7,9,11,19, 23,24 \};$
\item[(7)] $q=5^6$, $\delta=0.5$ and $r\in \{1,3,4,9,19,24,30,49,61 \};$
\item[(8)]  $q=5^8$, $\delta=0.5$ and $r\in \{1,2,3,4,5,7,9,11,12,15,19,23,24,25,38,47,49,51\}$.
\end{itemize}
The following figures \ref{fig:1} and  \ref{fig:2} show the bound given in Theorem \ref{thm:1.1} and the asymptotic Gilbert-Varshamov bound of locally repairable codes given for $r=2,q=3^6$ and $r=6,q=2^{12}$, respectively.

\begin{figure}
\centering
\includegraphics[width=3.5in]{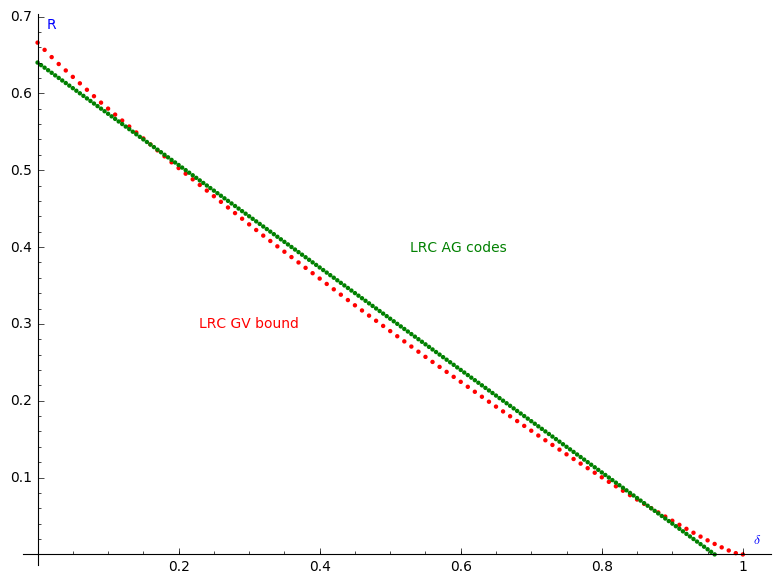}
\caption{$r=2, q=3^6$}\label{fig:1}
\end{figure}

\begin{figure}
\centering
\includegraphics[width=3.5in]{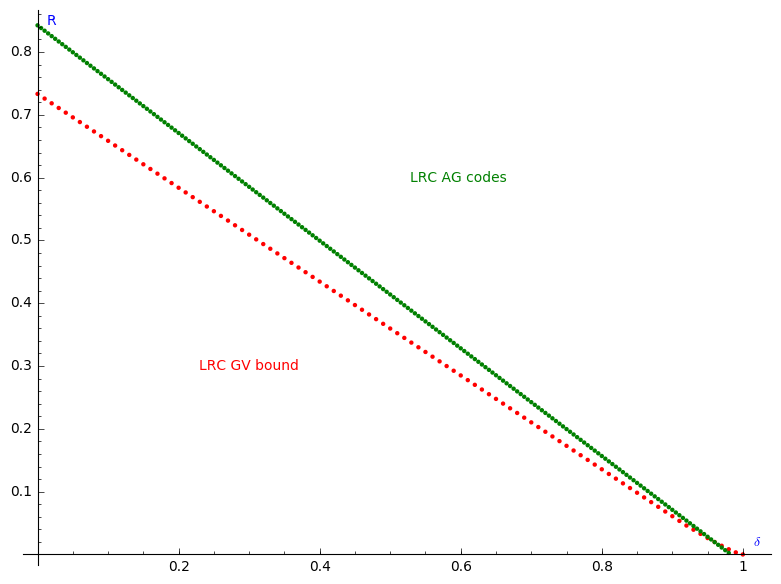}
\caption{$r=6, q=2^{12}$}\label{fig:2}
\end{figure}

Actually we can show that when $q$ is large enough, our bound \eqref{eq:main} given in Theorem \ref{thm:1.1} always exceeds the asymptotic Gilbert-Varshamov bound  of locally repairable codes  for some range of locality $r$.

\begin{prop}\label{prop:1.3} 
If the locality $r$ lies in the range $[c\log_2q, \frac{q}{2\log_2q}]$ for any constant $c>1$, then the bound \eqref{eq:main} given in Theorem \ref{thm:1.1} exceeds the asymptotic Gilbert-Varshamov bound \eqref{asymp_GV_bound} of locally repairable codes for all sufficiently large $q$.
\end{prop}
See Appendix for the proof of Proposition \ref{prop:1.3}.
\begin{rmk}  Indeed, our numerical result shows that if localities $r$ are too small, 
the bound \eqref{eq:main} given in Theorem \ref{thm:1.1} may not exceed the asymptotic Gilbert-Varshamov bound \eqref{asymp_GV_bound}. 
For example, the bound \eqref{eq:main} is worse than the asymptotic Gilbert-Varshamov bound \eqref{asymp_GV_bound} for $r=3, q=64$ or $r=2, q=81$.  On the other hand, it is easy to verify that if localities $r$ are too large, the asymptotic Gilbert-Varshamov bound \eqref{asymp_GV_bound} always outperforms the bound \eqref{eq:main} given in Theorem \ref{thm:1.1}.
%However, we could not prove this due to the  complicated Gilbert-Varshamov bound.
\end{rmk}

Finally in this section, we consider the naive construction given in \cite{CM15} and compare the bounds from this naive construction with various bounds discussed above.

\begin{lemma}\label{lem:1.4} (see \cite{CM15}) As long as there is a $q$-ary $[n,k,d]$-classical linear code, there exists a $q$-ary $[n,k-\frac{n}{r+1},d]$-locally repairable code with locality $r$ for any positive integer $r$ with $(r+1)|n$ and $r\ge \frac nk$.
\end{lemma}
 The \LRC with a locality of $r$ in the above lemma is obtained from an $[n,k,d]$-linear code $C$ by adding $n/(r+1)$ rows to a parity check matrix of the code $C$ such that each new row has exactly $r+1$ nonzero entries that are equal to $1$ and the support of all the new rows are disjoint.

Combining the classical Gilbert-Varshamov  bound with Lemma \ref{lem:1.4} gives an existence lower bound on \LRCs
\begin{equation}\label{eq:NGV}R_q(r,\delta) \ge 1-H_q(\delta)-\frac{1}{r+1}=\frac{r}{r+1}-H_q(\delta).\end{equation}
If we combine the Tsfasman-Vl\v{a}du\c{t}-Zink bound \cite[Theorem 8.4.7]{St09} with Lemma \ref{lem:1.4}, then we have
\begin{equation}\label{asymp_tvz_bound}
R_q(r,\delta) \ge \frac{r}{r+1}-\delta-\frac{1}{\sqrt{q}-1}
\end{equation}
for any square prime power $q$.

One can verify that the bound \eqref{eq:NGV} is worse than the the asymptotic Gilbert-Varshamov bound \eqref{asymp_GV_bound} on locally repairable codes. 
However, the bound given in \eqref{asymp_tvz_bound} is better than the bound given in \eqref{construction_l} for $q\ge 43$.

Furthermore,   the bound given in \eqref{asymp_tvz_bound} is worse than the bound given in \eqref{construction_l+1} for $\delta > \frac{r(r-1)}{q-1}-\frac{1}{\sqrt{q}-1}$. It is also worse than our bound \eqref{eq:main} for $$\delta> \frac{r(r-1)-\sqrt{q}}{q-\sqrt{q}}.$$
 Thus, if the locality $r\le q^{1/4}$, then the bound given in \eqref{asymp_tvz_bound} is worse than both the bound given in \eqref{construction_l+1} and our bound \eqref{eq:main} for all $\delta$. On the other hand, if $r>\sqrt{q}$, then the bound given in \eqref{asymp_tvz_bound} is the best among three. This means that the bound from the naive construction is quite good for relatively large locality.

\section{Preliminaries}\label{sec:2}
In this section, we present some preliminaries on function fields over finite fields, algebraic geometry codes, and the asymptotically optimal Garcia-Stichtenoth function field tower given in \cite{GS96}.

\subsection{Function fields over finite fields}
Let $F/\F_q$ be a function field with the full constant field $\F_q$.
Let $\mathbb{P}_F$ denote by the set of places of $F$ and let  $g(F)$ denote by the genus of $F$.
The principal divisor of $z\in F^*$ is defined by $$(z)=\sum_{P\in \mathbb{P}_F} v_P(z)P$$ where $v_P$ is the normalized discrete valuation with respect to the place $P$.
Let $G$ be a divisor of $F$.
The Riemann-Roch space $$\mathcal{L}(G)=\{z\in F^*: (z)\ge -G\}\cup \{0\}$$
is a finite dimensional vector space over $\F_q$ and its dimension is $\ell(G)\ge \deg(G)-g(F)+1$ from Riemann's theorem  \cite[Theorem 1.4.17]{St09}.
Let $E$ be a subfield of $F$ with the same full constant field $\F_q$. Then
the Hurwitz genus formula yields
$$2g(F)-2=[F:E](2g(E)-2)+\deg \text{Diff}(F/E),$$
where Diff$(F/E)$ stands for the different of $F/E$ \cite[Theorem 3.4.13]{St09}. Since the different of $F/E$ is an effective divisor of $F$,
one has
$$2g(F)-2\ge [F:E](2g(E)-2).$$

Let $\Aut(F/\F_q)$ be the automorphism group of the function field $F$ over $\F_q$, that is,
$$\Aut(F/\F_q)=\{\sigma: F\mapsto F \ \mid \  \sigma \text{ is an } \F_q \text{-automorphism of } F\}.$$
Now we consider the group action of the automorphism group $\Aut(F/\F_q)$ on the set of places $\mathbb{P}_F$.
For any automorphism $\Gs\in \Aut(F/\F_q)$ and any place $P\in \mathbb{P}_F$, then $\Gs(P)=\{\Gs(z):z\in P\}$ is a place of $F$ as well.
The stabilizer of the place $P$ under the action of the automorphism group $\Aut(F/\F_q)$  is called the decomposition group of $P$ and
denoted by $$\mG(P)=\{\Gs\in \Aut(F/\F_q): \Gs(P)=P\}.$$
The orbit of $P$ under the action of $\Aut(F/\F_q)$ is denoted by $[P]$ and given by $$[P]=\{\Gs(P):\Gs\in \Aut(F/\F_q)\}.$$
Then the size of  $\Aut(F/\F_q)$ is $|\Aut(F/\F_q)|=|\mG(P)|\cdot |[P]|.$

Let $\mG$ be a subgroup of $\Aut(F/\F_q)$. The fixed subfield of $F$ with respect to $\mG$ is defined by
$$F^{\mG}=\{z\in F: \Gs(z)=z \text{ for all } \Gs\in \mG\}.$$
From the Galois theory, $F/F^{\mG}$ is a Galois extension with $\Gal(F/F^{\mG})=\mG.$
For any place $P\in \mathbb{P}_F$, the place $P\cap F^{\mG}$ is splitting completely in $F$ if and only if $\Gs(P)$ are pairwise distinct for all automorphisms $\Gs\in \mG$.

\subsection{Algebraic geometry codes}
In the subsection, we introduce the general construction of algebraic geometry codes.
The reader may refer to \cite{NX01,TV91,TVN90} for details.
Let $F/\F_q$ be a function field over the full constant field $\F_q$. Let $\mP=\{P_1,\dots,P_n\}$ be a set of $n$ distinct rational places of $F$. For a divisor $G$ of $F$ with $0<\deg(G)<n$ and $\supp(G)\cap\mP=\emptyset$, the algebraic geometry code is defined to be
\begin{equation}\label{eq:8}
C(\mP,G):=\{(f(P_1),\dots,f(P_n)): \; f\in\mL(G)\}.
\end{equation}
Then $C(\mP,G)$ is an $[n,k,d]$-linear code with dimension $k=\ell(G)$ and minimum distance $d\ge n-\deg(G)$.
If $V$ is a subspace of $\mL(G)$, we can define a subcode of $C(\mP,G)$ by
\begin{equation}\label{eq:9}
C(\mP,V):=\{(f(P_1),\dots,f(P_n)):\; f\in  V\}.
\end{equation}
Then the dimension of  $C(\mP,V)$ is the dimension of the vector space $V$ over $\F_q$ and the minimum distance of  $C(\mP,V)$ is still lower bounded by $n-\deg(G)$.

\subsection{Garcia-Stichtenoth function field tower}\label{asymtower}
Let  $q=\ell^2$ be a square of a prime power.
In this subsection, we consider the asymptotically optimal Garcia-Stichtenoth function field tower $\mathcal{T}=(T_1,T_2,T_3,\cdots)$ which is  given by $T_m=\F_{q}(y_1,y_2,\cdots,y_m)$ with
\begin{equation}\label{optimaltower}
y_{i+1}^{\ell}+y_{i+1}=\frac{y_i^{\ell}}{y_i^{\ell -1}+1}, \mbox{ for } i=1,2,\cdots, m-1
\end{equation}
in  \cite{GS96}.
We summarize the main properties of the tower $\mathcal{T}=(T_1,T_2,T_3,\cdots)$ in the following proposition.

\begin{prop}\label{prop:2.1} 
\begin{itemize}
\item[(i)] $[T_m:\F_q(y_i)]=\ell^{m-1}$, for $i=1,\cdots,m.$
\item[(ii)] Let $P\in \mathbb{P}_{T_m}$ be a pole of $y_1$ or a zero of $y_1-\Ga$ for $\Ga\in \F_q$ with $\Ga^{\ell-1}=-1$. Then the rational place $P$ is a common pole of $y_2,y_3,\cdots,y_m$. Moreover, $P$ is totally ramified in all extensions $T_m/T_1$.
\item[(iii)] Let $Q_\Ga$ denote the zero of $y_1-\Ga$ in $T_1$ for $\Ga\in \F_q$.
Then any rational place  $Q_\Ga$  with $\Ga^\ell+\Ga\neq 0$ splits completely in all extensions $T_m/T_1$. Hence, the number of rational places of $T_m$ is $N(T_m)\ge (q-\ell)\ell^{m-1}+\ell.$
\item[(iv)]  The genus of $T_m$ is given by
$$g(T_m)=\begin{cases}
(\ell^{\frac{m}{2}}-1)^2, & \mbox{ if } m \equiv 0(\mbox{mod } 2),\\
(\ell^{\frac{m+1}{2}}-1)(\ell^{\frac{m-1}{2}}-1), & \mbox{ if } m \equiv 1(\mbox{mod } 2).
\end{cases}$$
\item[(v)]  The tower  $\mathcal{T}=(T_1,T_2,T_3,\cdots)$ is asymptotically optimal, since it attains the Drinfeld-Vl\^{a}du\c{t} bound over $\F_q$, i.e., 
$$\limsup_{m\rightarrow \infty} \frac{N(T_m)}{g(T_m)}=\sqrt{q}-1.$$
\end{itemize}
\end{prop}

Let $\mP$ denote by the set of places of $T_m$ lying over $Q_\Ga$ for all $\Ga\in \F_q$ with $\Ga^\ell+\Ga\neq 0.$
Then the cardinality of $\mP$ is $\ell^{m-1}(q-\ell)$ from Proposition \ref{prop:2.1}(iii). For any place $P\in \mP$, there exists $\Ga_1\in \F_q$ with $\Ga_1^\ell+\Ga_1\neq 0$ such that $P$ is a zero place of $y_1-\Ga_1$ in $T_m$. By induction, we can show that there exist $\Ga_i\in \F_q$ with $\Ga_i^\ell+\Ga_i=\Ga_{i-1}^{\ell}/(\Ga_{i-1}^{\ell-1}+1)$ such that $P$ is a zero place of $y_i-\Ga_i$ for $2\le i\le m$. Hence, $P$ is a common zero of $y_1-\Ga_1,\cdots,y_m-\Ga_m$ in $T_m$ and we can identify $P$ with the $n$-tuple $(\Ga_1,\Ga_2,\cdots,\Ga_m)$.

From Proposition \ref{prop:2.1}(ii), the infinity place $\infty$ of $T_1$ is totally ramified in all extensions $T_m/T_1$ for $m\ge 2$. Denote by $P_{\infty, m}$ the unique place of $T_m$ lying over $\infty$. 
It can be shown that $v_{P_{\infty,m}}(y_m)=-1$ for any positive integer $m$ from the theory of Artin-Schreier extensions.
It seems to be difficult to determine the automorphism group of $T_m$ over $\F_q$, since it is not easy to determine the orbit of $P_{\infty,m}$ (see \cite{GS96,La13,PST98}). Fortunately,   the decomposition group of $P_{\infty, m}$
$$\mG(P_{\infty, m})=\{\Gs\in \Aut(T_m/\F_q):\Gs(P_{\infty, m})=P_{\infty, m}\}$$
was determined explicitly in  \cite{La13}. 
In particular, the decomposition group of $P_{\infty, m}$ consists of all automorphisms $\Gs$ with the following form
$$\begin{cases}
\Gs(y_i) = cy_i  \text{ for } i=1,\cdots, m-1,\\
\Gs(y_m) = cy_m+a,
\end{cases} $$
where $c\in \F_\ell^*$ and $a^\ell+a=0$ for odd $q$ or $m=1$.
For even $q$ and $m\ge 2$, the decomposition group of $P_{\infty, m}$ consists of all automorphisms $\Gs$ with
$$
\begin{cases}
\Gs(y_i)=cy_i \text{ for } 1\le i\le m-2,\\
\Gs(y_{m-1})=cy_{m-1}+b,\\
\Gs(y_m)=cy_m+\frac{b^2}{cy_{m-2}}+a,
\end{cases}
$$
where $b\in \F_\ell$, $a^\ell+a=b, \text{ and } c\in \F_\ell^*.$ 

Irrespective of the characteristic of $\F_q$, let $\mA$ denote by the set of all automorphisms $\Gs$ with 
$$\begin{cases}
\Gs(y_i) = cy_i  \text{ for } i=1,\cdots, m-1,\\
\Gs(y_m) = cy_m+a,
\end{cases} $$
where $c\in \F_{\ell}^*$ and $a^\ell+a=0$.
It is easy to verify that $\mA$ is a subgroup of the decomposition group of $P_{\infty, m}$. Let $P=(\Ga_1,\Ga_2,\cdots,\Ga_m)$ be a rational place of $\mP$. 
By considering the action of automorphism group on the set of places, we have
 $$y_m(\Gs^{-1}(P))=(\Gs(y_m))(P)=(cy_m+a)(P)=cy_m(P)+a=c\Ga_m+a.$$
Hence, it is easy to verify that
\begin{equation}\label{sigmaac}
\sigma^{-1}(P)=(c\Ga_1,\cdots,c\Ga_{m-1}, c\Ga_{m}+a).
\end{equation}
In the following section, we will focus on the subgroups of $\mA$.

\section{Construction of asymptotically good locally repairable codes}
The main propose of this section is to prove Theorem \ref{thm:1.1}.
We first introduce a general construction of locally repairable codes from automorphism groups of function fields and then employ this method to construct families of locally repairable codes with good asymptotic parameters from the Garcia-Stichtenoth function field tower $\mathcal{T}=(T_1,T_2,T_3,\cdots)$.

\subsection{Construction of locally repairable codes via automorphism groups}
In this subsection, we present a general construction of locally repairable codes from automorphism groups of function fields. 
This method was initiated systematically to construct optimal locally repairable codes from automorphism groups of rational function fields \cite{JMX17,TB14}.

Let $F/\F_q$ be a function field over the full constant field $\F_q$. Let $\Aut(F/\F_q)$ be the automorphism group of $F$ over $\F_q$.
Let $\mG$ be a subgroup of $\Aut(F/\F_q)$ of order $r+1$ and let $F^{\mG}$ be the fixed subfield of $F$ with respect to $\mG$.
Denote by $E$ the fixed subfield $F^{\mG}$. Then $F/E$ is a Galois extension with Galois group $\Gal(F/E)=\mG.$

Assume that there exist at least one rational place of $F$ which is ramified in $F/E$ and $m$ rational places $Q_1,\cdots, Q_m$ of $E$ which are all splitting completely in $F/E$.
Let $P_{i,1}, P_{i,2},\cdots, P_{i,r+1}$ be the $r+1$ rational places of $F$ lying over $Q_i$ for each $1\le i \le m$.
Put $\mP=\{P_{i,j}: 1\le i \le m, 1\le j\le r+1\}$. Then the cardinality of $\mP$ is $m(r+1)$.

Choose a divisor $G$ of $E$ such that $\text{supp}(G) \cap \{Q_1,\cdots,Q_m\}=\emptyset$.
The Riemann-Roch space $\mathcal{L}(G)=\{f\in E^*: (f)\ge -G\}\cup \{0\}$
 is a finite dimensional vector space over $\F_q$ with dimension $\ell(G)\ge \deg(G)-g(E)+1$, where $g(E)$ is the genus of $E$.
Let $\{z_1,\cdots, z_t\}$ be a basis of the Riemann-Roch space $\mathcal{L}(G)$ over $\F_q$.
Choose an element $x\in F$ such that $1,x,\cdots,x^{r-1}$ are linearly independent over $E$ and  $x(P_{ij})$ are pairwise distinct at the rational places $P_{i,1}, P_{i,2},\cdots, P_{i,r+1}$ for each $1\le i \le m$.
Consider the set of functions
$$V:=\Big{\{}\sum_{i=0}^{r-1} \Big{(}\sum_{j=1}^t a_{ij}z_j\Big{)}x^i\in F: a_{ij}\in \F_q \Big{\}}.$$

\begin{theorem}\label{thm:3.1} 
Let $\mP$ and $V$ be defined as above, then the algebraic geometry code
$$C(\mP,V)=\{(f(P))_{P\in \mP}: f\in V\}$$ is a $q$-ary $[n,k,d]$-locally repairable code with locality $r$, length $n=m(r+1)$, dimension $k=rt \ge r(\deg(G)-g(E)+1)$ and minimum distance $d\ge n-(r+1)\deg(G)-(r-1)\deg(x)_{\infty}$.
\end{theorem}

\begin{proof}
First, it is easy to see that the dimension of $V$ over $\F_q$ is $rt=r\ell(G)$,
since $1,x,\cdots, x^{r-1}$ are linearly independent over $E$ and $\{z_1,\cdots, z_t\}$ is a basis of the Riemann-Roch space $\mathcal{L}(G)$ over $\F_q$.
The elements $z_jx^i$ with $z_j\in \mL(G)$ have at most $(r+1)\deg(G)+(r-1)\deg(x)_{\infty}$ zeros for $0\le i\le r-1$ and $1\le j\le t$. Hence, the minimum distance of $C(\mP,V)$ is lower bounded by $d\ge n-(r+1)\deg(G)-(r-1)\deg(x)_{\infty}$. Under the assumption that $d\ge 1$, the dimension of  $C(\mP,V)$ is $k=rt \ge r(\deg(G)-g(E)+1)$ from Riemann's Theorem.
The locality property follows from the Lagrange interpolation formula, since $x$ takes pairwise distinct values on the set of $r$ rational places 
$\{\Gs(P): 1\neq \Gs\in \mG\}$ for any rational place $P\in \mP$.
\end{proof}

\begin{rmk}
The result of Theorem \ref{thm:3.1} can be found in  \cite[Theorem 3.1]{BTV17} and \cite[Proposition 3.1]{JMX17} for the rational function field $F$ over $\F_q$.
\end{rmk}

\subsection{Locally repairable codes from asymptotic optimal towers}
In this subsection, we will construct locally repairable codes via the subgroups of the automorphism group of the asymptotic optimal 
Garcia-Stichtenoth function field tower $\mathcal{T}=(T_1,T_2,T_3,\cdots)$.
Let $\mA$ denote by the set of all automorphisms $\Gs$ of $T_m$ with 
$$\begin{cases}
\Gs(y_i) = cy_i  \text{ for } i=1,\cdots, m-1,\\
\Gs(y_m) = cy_m+a,
\end{cases} $$
where $c\in \F_{\ell}^*$ and $a^\ell+a=0$.
As mentioned in Section \ref{asymtower}, $\mA$ is a subgroup of the decomposition group of $P_{\infty, m}$. 
In fact, we mainly focus on the subgroups of $\mA$, especially on the orders of the subgroups.
The structure of subgroups of $\mA$ can be determined from the proof of the following proposition. 

\begin{prop} \label{prop:3.2}
Let $\ell=p^w$ be a prime power.
Let $v$ be an integer with $0\le v\le w$ and let $u$ be a positive integer satisfying $u|\gcd(p^v-1, \ell-1)$.
Then there is a subgroup $\mG$ of  $\mA$ of order $up^v$.
\end{prop}
\begin{proof}
If $u$ is a divisor of $\ell-1$, then there exists a subgroup $H$ of the multiplicative group $\F_\ell^*$ of order $u$. As $u|(p^v-1)$,  the field $\F_p(H)$ is contained in $\F_{p^v}$. Put $h=\min \{ t>0: u|(p^t-1)\}$. Then we have
$\F_p(H)=\F_{p^h}\le \F_\ell$ and $h|\gcd(v,w)$.

Let $W$ be a vector subspace of $\{a\in \F_q: a^\ell+a=0\}$ over $\F_{p^h}$ with dimension $v/h$. Let $\mG$ be the set of all automorphisms $\Gs$ with the form
$$ \Gs(y_i)=cy_i \text{ for } 1\le i\le m-1 \text{ and } \Gs(y_m)=cy_m+a,$$
where $c\in H$ and $a\in W$.
Then it is easy to verify that $\mG$ is a subgroup of  $\mA$  of order $up^{v}$.
\end{proof}

Let $\mP$ be a subset of $\mathbb{P}_{T_m}$ with
$$\mP=\Big{\{}(\Ga_1,\Ga_2,\cdots,\Ga_m):\Ga_1^{\ell}+\Ga_1\neq 0, \Ga_i^\ell+\Ga_i=\frac{\Ga_{i-1}^{\ell}}{\Ga_{i-1}^{\ell-1}+1} \text{ for } 2\le i\le m\Big{\}}.$$
Now we can construct locally repairable codes with good asymptotic parameters from the subgroups of $\mA$ which is a subgroup of the decomposition group of $P_{\infty,m}$.

\begin{theorem}\label{thm:3.4}
Let $q=\ell^2$ with $\ell=p^w$ and let $r$ be a positive integer. For any integer $v$ with $0\le v\le w$ and any positive integer $u$ satisfying 
$u|\gcd(p^v-1,\ell-1)$, let $up^v=r+1$.
Then there exists a family of $q$-ary $[n_m,k_m,d_m]$-locally repairable codes with locality $r$ constructed from $T_m$ with the parameters
$$\begin{cases}
n_m= \ell^{m-1}(q- \ell),\\
k_m \ge r s-r(g(T_m)-1)/(r+1),\\
d_m \ge n_m- (r+1) s-(r-1)\ell^{m-1},
\end{cases}$$
where $(g(T_m)-1)/(r+1)  \le s\le n_{m-1}.$
\end{theorem}

\begin{proof}
Let $\mG$ be any subgroup of $\mA$ of order $up^v=r+1$.
Then $P_{\infty,m}$ is totally ramified in $T_m/T_m^{\mG}$.
Let $P=(\Ga_1,\Ga_2,\cdots,\Ga_m)$ be a rational place of $\mP$. For any automorphism $\Gs\in \mG$ with $\Gs(y_m)=cy_m+a$, we have
$$\sigma^{-1}(P)=(c\Ga_1,\cdots,c\Ga_{m-1}, c\Ga_{m}+a)$$
from Equation \eqref{sigmaac}.
Then the rational places $\Gs^{-1}(P)$ are pairwise distinct for all automorphisms $\Gs\in \mG$. 
Hence, $P\cap T_m^{\mG}$ is splitting completely in the extension $T_m/T_m^{\mG}$ for any place $P\in \mP$.

Let $x=y_m$.
We claim that $x$ takes pairwise distinct values on the set of $r+1$ rational places  $\{\Gs^{-1}(P):\Gs\in \mG\}$ for each $P\in \mP$.
Let $\Gs_1$ and $\Gs_2$ be two different automorphisms of $\mG$ with $\Gs_i(y_m)=c_iy_m+a_i$, where $c_i\in \F_\ell^*$ and $a_i^\ell+a_i=0$ 
 for $i=1,2$.
It follows that $y_m(\Gs_i^{-1}(P))=c_i\Ga_m+a_i$ for $i=1,2.$
Suppose that $c_1\Ga_m+a_1=c_2\Ga_m+a_2$.
If $c_1=c_2$, then $a_1=a_2$, i.e., $\Gs_1=\Gs_2$. Hence, $c_1\neq c_2$ and $\Ga_m=(a_2-a_1)/(c_1-c_2)$. It is easy to calculate that $$\Ga_m^\ell+\Ga_m=\Big{(}\frac{a_2-a_1}{c_1-c_2}\Big{)}^\ell+\frac{a_2-a_1}{c_1-c_2}=\frac{(a_2^\ell+a_2)-(a_1^\ell+a_1)}{c_1-c_2}=0,$$
which is a contradiction to $(\Ga_1,\cdots,\Ga_m)\in \mP$. Thus, $y_m(\sigma_1^{-1}(P))\neq y_m(\sigma_2^{-1}(P))$.

Let $Q_{\infty,m}=P_{\infty,m} \cap T_m^{\mG}$. Choose $G=s Q_{\infty,m}$ and the dimension of the Riemann-Roch space $\mathcal{L}(G)$ over $\F_q$ is $\ell(G)\ge s-g(T_m^{\mG})+1.$
The Hurwitz genus formula yields the inequality
$2g(T_m)-2\geq (r+1)[2g(T_m^{\mG})-2].$
Then we have $$\ell(G) \ge s- \frac{g(T_m)-1}{r+1}.$$
Let $\{z_1,\cdots, z_t\}$ be a basis of the Riemann-Roch space $\mathcal{L}(G)$ over $\F_q$.
Consider the set of functions
$$V=\Big{\{}\sum_{i=0}^{r-1} \Big{(}\sum_{j=1}^t a_{ij}z_j\Big{)}x^i\in T_m: a_{ij}\in \F_q \Big{\}}.$$
We claim that $1,x,\cdots,x^{r-1}$ are linearly independent over $T_m^{\mG}$.
Suppose that the claim is false, i.e., there exist $c_i\in T_m^{\mG}$ for $0\le i \le r-1$ such that
$$\sum_{i=0}^{r-1} c_ix^i=0,  \text{ not all } c_i=0.$$
For $c_i\neq 0$, we obtain
$$v_{P_{\infty,m}}(c_ix^i)=v_{P_{\infty,m}}(c_i)+i\cdot v_{P_{\infty,m}}(x)= (r+1)v_{Q_{\infty,m}}(c_i)-i\equiv -i \ (\text{mod } r+1),$$
from the facts $P_{\infty,m}$ is totally ramified in $T_m/T_m^{\mG}$ and $v_{P_{\infty,m}}(x)=v_{P_{\infty,m}}(y_m)=-1.$
Therefore $v_{P_{\infty,m}}(c_ix^i)\neq v_{P_{\infty,m}}(c_jx^j)$ whenever $i\neq j$, $c_i\neq 0$ and $c_j\neq 0$. The Strict Triangle Inequality \cite[Lemma 1.1.11]{St09} yields
$$v_{P_{\infty,m}}\Big{(}\sum_{i=0}^{r-1} c_ix^i\Big{)}=\min \Big{\{}  v_{P_{\infty,m}}(c_ix^i):c_i\neq 0 \Big{\}}<\infty,$$
which is a contradiction. Thus $1,x,\cdots,x^{r-1}$ are linearly independent over $T_m^{\mG}$.

From Theorem \ref{thm:3.1}, the dimension of algebraic geometry code $C(\mP,V)$ is
$$k_m=rt=r\ell(G)\ge rs-\frac{r}{r+1}[g(T_m)-1]$$
and the minimum distance $d_m$ of $C(\mP,V)$ is lower bounded by 
$$d_m \ge n_m- (r+1)s-(r-1)\deg (x)_{\infty} =n_m- (r+1)s-(r-1)\ell^{m-1}.$$
The last equality holds true, since the degree of the pole divisor of $x$ in $T_m$ is $$\deg(x)_{\infty}=\deg(y_m)_{\infty}=[T_m:\F_q(y_m)]=\ell^{m-1}$$
from Proposition \ref{prop:2.1}(i).
\end{proof}

\subsection{Proof of Theorem \ref{thm:1.1}}
With all preparations in this section, we are now able to prove Theorem \ref{thm:1.1}.
\begin{proof} Let $\{C_m=[n_m,k_m,d_m]_q\}$ be the family of \LRCs with locality $r$ constructed in Proposition \ref{prop:3.2}. 
An easy computation shows that
\begin{eqnarray*}
d_m+\frac{r+1}{r} k_m  &\ge &  n_m- (r+1)s-(r-1)\ell^{m-1}+(r+1)s-(g(T_{m})-1) \\
& \ge & n_m-(r-1)\ell^{m-1}-g(T_{m})+1 \\
& \ge & n_m-(r-1)\ell^{m-1}-\ell^m
\end{eqnarray*}
from Theorem \ref{thm:3.4} and Proposition \ref{prop:2.1}(iv).
Putting $\delta=\lim_{m\rightarrow \infty} \frac{d_m}{n_m}$ and $R=\lim_{m\rightarrow \infty} \frac{k_m}{n_m}$, we obtain
$$\delta+ \frac{r}{1+r}R\ge 1 -\frac{r-1}{q-\ell}-\frac{\ell}{q-\ell}.$$
The desired result  follows immediately.
\end{proof}

\section*{Appendix}
We provide a proof for Proposition \ref{prop:1.3} in this appendix.
\begin{proof}
Put $\delta=1/2$ and $$h(s)=\frac{1}{r+1} \log_q((1+(q-1)s)^{r+1}+(q-1)(1-s)^{r+1})-\delta \log_q(s).$$
Then the derivative
$$h^{\prime}(s)=\frac{(1+(q-1)s)^r((q-1)s-1)-(q-1)(1-s)^r(1+s)}{2s((1+(q-1)s)^{r+1}+(q-1)(1-s)^{r+1})\cdot \ln(q)} $$
is increasing in the interval $(0,1]$ and has a unique critical point $s_0\in (0,1]$ such that $h^{\prime}(s_0)=0.$
It follows that $h(s)$ is decreasing in the interval $(0,s_0]$ and increasing in the interval $[s_0,1]$. Hence, $h(s)$ achieves the minimum value at the point $s=s_0$. It is easy to verify that $s_0\in (\frac{1}{q-1},\frac{1}{q-1}+\epsilon)$ with $\epsilon=2^{-r}$ from the derivation $h^{\prime}(s)$.
From the mean value theorem, there exists $s_1\in (\frac{1}{q-1},s_0)$ such that
$$h(s_0)= h\Big{(}\frac{1}{q-1}\Big{)}+h^{\prime}(s_1)\Big{(}s_0-\frac{1}{q-1}\Big{)}\ge h\Big{(}\frac{1}{q-1}\Big{)}-\frac{q \epsilon}{\ln q}.$$
Hence, for $r\in[c\log_2 q, \frac{q}{2\log_q}]$ with $c>1$,
\begin{eqnarray*}
\min_{0<s\le 1}h(s) &=& h(s_0) \\
%\ge h\Big{(}\frac{1}{q-1}\Big{)}-\frac{q\cdot 2^{-r}}{\ln(q)}\\
&\ge& \frac{1}{r+1}\log_q\Big{(}2^{r+1}+(q-1)\Big{(}\frac{q-2}{q-1}\Big{)}^{r+1}\Big{)}-\delta\log_q\Big{(}\frac{1}{q-1}\Big{)}-\frac{q\cdot 2^{-r}}{\ln q}\\
&\ge & \log_q2+\delta+\delta\log_q\Big{(}1-\frac{1}{q}\Big{)}-\frac{1}{q^{c-1}\ln q}\\
&\ge & \frac{1}{r+1}+ \frac{r}{r+1}\delta+ \frac{r}{r+1}\frac{\sqrt{q}+r-1}{q-\sqrt{q}},
\end{eqnarray*}
provided that $q$ is sufficiently large.
\end{proof}

\end{document}